\documentclass[aps,prl,twocolumn,superscriptaddress]{revtex4-1}

\usepackage{amsmath}
\usepackage{amsfonts}
\usepackage{amssymb}
\usepackage{amsthm}
\usepackage{bm}
\usepackage{xparse}
\usepackage{graphicx}
\usepackage{xcolor}
\usepackage{textcase}
\usepackage{url}
\usepackage{lipsum}

\newcommand{\avg}[1]{\langle{#1}\rangle} % for averages
\newcommand{\ket}[1]{| {#1} \rangle} % for Dirac bras
\newcommand{\bra}[1]{\langle {#1} |} % for Dirac kets
\DeclareDocumentCommand{\Tr}{m m O{\big}}{{\rm Tr}_{\:\!{#1}}#3({#2}#3)}

\DeclareMathOperator*{\tr}{\mathrm{Tr}}
\DeclareMathOperator*{\Var}{\mathrm{Var}}

\newtheorem{theorem}{Theorem}

\begin{document}
\title{On the central limit theorem for unsharp quantum random variables}
\author{Aleksandra Dimi\'{c}}
\affiliation{
Faculty of Physics,
University of Belgrade,
Studentski Trg 12-16,
11000 Belgrade, Serbia
}
\author{Borivoje Daki\'c}
\affiliation{
Institute for Quantum Optics and Quantum Information (IQOQI),
Austrian Academy of Sciences, Boltzmanngasse 3,
A-1090 Vienna, Austria
}

\date{\today}

\begin{abstract}

In this letter we study the weak-convergence properties of random variables generated by unsharp quantum measurements. More precisely, for a sequence of random variables generated by repeated unsharp quantum measurements, we study the limit distribution of relative frequency. We provide the de Finetti-type representation theorem for all separable states, showing that the distribution can be well approximated by mixture of normal distributions. No symmetry restrictions, such as the permutational invariance were needed. Furthermore, we investigate the convergence rates and show that the relative frequency can stabilize to some constant at best at the rate of order $1/\sqrt N$ for all separable inputs. On the other hand, we provide an example of a strictly unsharp quantum measurement where the better rates are achieved by using entangled inputs. This means that in certain cases the noise generated by the measurement process can be suppressed by using entanglement. We deliver our result in the form of quantum information task where the player achieves the goal with certainty in the limiting case by using entangled inputs or fails with certainty by using separable inputs.
\end{abstract}

\maketitle

\emph{Introduction.}$-$ Quantum theory predicts probability distributions of measurement outcomes. In practice, we identify probabilities with the relative frequencies of measurement outcomes in the limit of large number of experimental runs. The identification is justified by the so-called \emph{i.i.d. assumption}, which demands that a certain physical process (e.g. use of quantum channel or state) is repeated arbitrarily many times \emph{identically and independently} of other processes. The convergence to probability is guaranteed by the weak law of large numbers and the errors are quantified by the central limit theorem (CLT). However, one can naturally ask what happens if the i.i.d. assumption no longer applies? Clearly, such a framework is much less structured and it opens-up new possibilities and imposes new limitations both for quantum foundations and quantum information processing~\cite{Renner0,Hastings,Yard}.

In classical probability theory the weak convergence properties are fairly well understood for independent variables. The pioneering works by  Kolmogorov, Chebyshev, Lindeberg and Lyaponov provided a good set of conditions for CLT and asymptotic normality to hold (see for example \cite{Billingsley,Dudley}). On the other hand, dependent variables are much more difficult to tackle. Generally, one has to impose certain restrictions, otherwise there is not much to say in the most general case. For example, the set of sufficient conditions for CLT to hold can be provided for the weakly-dependent variables~\cite{Bradley}. Furthermore, for random variables under the symmetry constraints such as exchangeability, one can derive the exact necessary and sufficient conditions~\cite{Chernoff}. In quantum setting, apart from the i.i.d. scenario~\cite{Hudson1,Hudson2,Hayashi,Jaksic}, in recent years we have seen a plethora of CLT-type results mainly in the context of quantum many-body dynamics~\cite{Eisert,Buchholz,Anshu, Brandao1,Brandao2,Arous}. Yet, the main focus in these studies are the properties of quantum states and observables without counting the measurement effects. However, these effects are unavoidable in general. Our main concern here is to investigate how is the noise generated by the measurement process affecting the limiting distribution of measured quantities, i.e. in the asymptotic limit of many repeated measurement runs.

Our basic idea is that noise and uncertainty produced by generalized quantum measurement (POVM) will resemble some form of asymptotic normality even for correlated inputs. Unlike the standard CLTs where the distribution of the relative frequency is approximated by a single normal distribution, we will deliver our results in the form of de Finetti representation theorem~\cite{de Finneti}, where the output distribution is well approximated by a mixture of normal distributions. We will provide our representation theorem for all separable states. Contrary to the standard de Finetti representations, both in classical~\cite{de Finneti,Diaconis,Diaconis1} and quantum~\cite{Hudson,Caves,Christandl,Renner,Harrow,Li,Chiribella} scenario, our result \emph{does not require any symmetry constraint}, such as the permutational invariance (exchangeability), as long as the inputs are subjected to strictly unsharp measurements. Furthermore, we will investigate the convergence rates and show that the relative frequency can stabilize to some constant value at best at the rate of order $1/\sqrt N$ for all separable inputs. On the other hand, we will provide a simple example where the entangled inputs can significantly increase the convergence rate.
%Therefore, in certain cases, entanglement can make the relative frequency to converge much faster compared to all separable states, consequently reducing the amount of noise generated by the measurement process. Based on this, we will device a scheme for the probabilistic multipartite entanglement detection by using only a single copy of the system. Unlike the conventional entanglement detection schemes~\cite{Horodecki,Guhne}, where one assumes repeated measurements on a very large number of copies of quantum system ( e.g. to measure some witness operator), we show that in certain cases the detection is possible by using only a single copy of quantum system. Our idea is similar to the recent proposals on entanglement detection from random correlations~\cite{Tran1,Tran2} where the entanglement is certified probabilistically, i.e. with the certain level of confidence.
We deliver our result in the form of ``quantum game''~\cite{Bell} where the player is able to accomplish the task with certainty in the asymptotic limit by using entangled inputs or fails with certainty by using any separable inputs.

\emph{Unsharp quantum measurements.}$-$ The basic feature of generalized quantum measurements (POVMs) is the ``unsharpness'' and production of an unavoidable noise during the measurement process~\cite{MassarS}. The noise comes due to a non-projective character of the measurement elements (operators), thus the measurement outcomes will necessarily fluctuate in the sequence of repeated experimental runs. Our main goal here is to show that the fluctuations can be very different depending on whether separable or entangled inputs are subjected to unsharp measurements.

To begin, we start with some basic definitions. Consider a quantum measurement defined by the set of POVM elements $E_i$, with $E_i\geq0$ and $\sum_iE_i=\openone$. We define a random variable $X$ generated by measurement with the set of numbers $X\in\{x_1,x_2,\dots\}$ where each $x_i\in\mathbb{R}$ corresponds to the $i$th outcome (defined by $E_i$). It is convenient to define the expectation $\hat{X}=\sum_ix_iE_i$ and the uncertainty operator $\Delta \hat{X}=\sum_ix_i^2E_i-\hat{X}^2$~\cite{MassarS}. For a given quantum state $\rho$, the expectation value and variance are easily evaluated, i.e. $\avg{\hat X}_{\rho}=\tr\rho\hat{X}$  and $\sigma^2=\Var_{\rho}[X]=\avg{\hat{X}^2}_{\rho}-\avg{\hat X}_{\rho}^2+\tr\rho\Delta\hat{X}$, respectively. We see that the uncertainty operator produces additional noise that comes solely due to measurement (note that $\Delta\hat{X}\geq0$ in general). For all projective (von Neumann) measurements $\Delta\hat{X}=0$, hence this term vanishes.

We focus on strictly unsharp measurements, that is we assume $\sigma_{-}\leq\sigma\leq\sigma_{+}$, with $\sigma_{-}>0$ being strictly positive for all states $\rho$. Furthermore, we assume that the third moment $r=\avg{|\hat X-\avg{\hat X}_{\rho}|^3}_{\rho}\leq M$ is bounded by some constant $M>0$ for all $\rho$.

For a sequence of random variables $X_1,\dots, X_N$ generated by repeated measurement, we set $X^{(N)}=X_1+\dots+X_N$ and $R_N=\frac{1}{N}X^{(N)}$ to be the relative frequency. Furthermore, we define the standardly normalized sum $S_N=\frac{1}{\sqrt N}(X^{(N)}-\avg{X^{(N)}})$. The distribution of the relative frequency $R_N$ is the central object of our investigation, i.e. what is the probability that $R_N$ takes some value in the limit of large number of experimental runs.

\emph{Separable inputs.}$-$ The answer to the previous question heavily depends on the type of input state. For example, if one supplies in each run the same state $\rho$, the overall input state is described by an i.i.d. state $\rho^{(N)}=\rho^{\otimes N}$, where $N$ is the number of experimental runs. The weak law guarantees the convergence of the relative frequency converges to the mean value $\avg{\hat X}_{\rho}$ and the central limit theorem states that the distribution of $S_N$ converges to the standard normal distribution. A slightly more delicate example is the one of independent inputs, i.e. $\rho^{(N)}=\rho_{1}\otimes\dots\otimes \rho_{N}$, where $\rho_i$s are different in general. Here we can define the mean variance
\begin{equation}
\Sigma_N^2=\frac{1}{N}\sum_{i=1}^{N}\sigma_i^2,
\end{equation}
with $\sigma_i^2=\Var_{\rho_i}[X]$ and the average mean
\begin{equation}
\mu=\frac{1}{N}\sum_{i=1}^{N}\mu_i,
\end{equation}
with $\mu_i=\avg{\hat X}_{\rho_i}$. Clearly $\avg{X^{(N)}}=N\mu$ and $\sigma_{-}\leq\Sigma_N\leq\sigma_{+}$ as each individual variance is bounded. We can apply the \emph{Lindeberg's condition} for CLT~\cite{Billingsley}, i.e.
\begin{equation}
\max_i\frac{\sigma_i^2}{\sum_{j=1}^N\sigma_j^2}=\max_i\frac{\sigma_i^2}{N\Sigma_N^2}\leq\frac{\sigma_{+}^2}{ N\sigma_{-}^2}\rightarrow0,
\end{equation}
when $N\rightarrow+\infty$, therefore the normalized sum
\begin{equation}
\frac{X^{(N)}-\avg{X^{(N)}}}{(\sum_{i=1}^{N}\sigma_i^2)^{1/2}}=\frac{X^{(N)}-N\mu}{\sqrt N\Sigma_N}=\frac{S_N}{\Sigma_N}
\end{equation}
converges to the standard normal distribution. To quantify the deviation for finite $N$, we can use the Berry-Esseen theorem~\cite{Berry,Esseen}. Let $P[S_N/\Sigma_N\leq x]$ be the cumulative distribution function (CDF) and $\Phi(x)$ is the CDF of the standard normal distribution, i.e. $\Phi(x)=1/\sqrt{2\pi}\int_{-\infty}^{x}e^{-t^2/2}dt$. We have
\begin{eqnarray}\label{BEbound}\nonumber
\sup_{x\in\mathbb{R}}|P[S_N/\Sigma_N\leq x]-\Phi(x)|&\leq& C_0\frac{\sum_{i=1}^{N}r_i}{N^{3/2}\Sigma_N^3}\\\nonumber
&\leq& C_0\frac{N M}{N^{3/2}\sigma_{-}^3}=\frac{C_0M}{\sigma_{-}^3\sqrt N},\\
\end{eqnarray}
where $r_i=\avg{|\hat X-\avg{\hat X}_{\rho_i}|^3}_{\rho_i}\leq M$ and $C_0$ is an absolute constant. We see that any product input state is subjected to CLT because the measurements are strictly unsharp (the variance is strictly bounded from bellow by $\sigma_{-}$). From here, we are ready to establish the representation theorem for separable states. For a given separable input state $\rho^{(N)}=\sum_k\lambda_k\rho_k^{(N)}$, where $\rho_k^{(N)}=\rho_{1,k}\otimes\dots\otimes \rho_{N,k}$, we set $\mu_{N,k}=\frac{1}{N}\sum_{i=1}^{N}\mu_{i,k}$,
with $\mu_{i,k}=\avg{\hat X}_{\rho_{i,k}}$ and $\Sigma_{Nk}^2=\frac{1}{N}\sum_{i=1}^N\sigma_{i,k}^2$. Here $\sigma_{i,k}^2=\mathrm{Var}_{\rho_{i,k}}[X]$.
\begin{theorem}
The CDF $F_N(x)=P[R_N\leq x]$ of the relative frequency satisfies the following bound:
\begin{equation}\label{RTheorem}
\sup_{x\in\mathbb{R}}\left|F_N(x)-\sum_k\lambda_k\Phi\left(\frac{x-\mu_{N,k}}{\Sigma_{Nk}/\sqrt{N}}\right)\right|\leq \frac{C_0M}{\sigma_{-}^3\sqrt N}.
\end{equation}
\end{theorem}

\begin{proof}{\footnotesize
Firstly, note that $F_N(x)=P[R_N\leq x]=\sum_k\lambda_kP_k[R_N\leq x]$, where $P_k[R_N\leq x]$ is the CDF for the product state $\rho_k^{(N)}=\rho_{1,k}\otimes\dots\otimes \rho_{N,k}$. We have
\begin{eqnarray}\nonumber
&&\sup_{x\in\mathbb{R}}\left|F_N(x)-\sum_k\lambda_k\Phi\left(\frac{x-\mu_{N,k}}{\Sigma_{Nk}/\sqrt{N}}\right)\right|\\\nonumber
&=&\sup_{x\in\mathbb{R}}\left|P\left[\frac{1}{N}X^{(N)}\leq x\right]-\sum_k\lambda_k\Phi\left(\frac{x-\mu_{N,k}}{\Sigma_{Nk}/\sqrt{N}}\right)\right|\\\nonumber
&=&\sup_{x\in\mathbb{R}}\left|\sum_k\lambda_k\left(P_k\left[\frac{1}{N}X^{(N)}\leq x\right]-\Phi\left(\frac{x-\mu_{N,k}}{\Sigma_{Nk}/\sqrt{N}}\right)\right)\right|\\\nonumber
&\leq&\sup_{x\in\mathbb{R}}\sum_k\lambda_k\left|P_k\left[\frac{1}{N}X^{(N)}\leq x\right]-\Phi\left(\frac{x-\mu_{N,k}}{\Sigma_{Nk}/\sqrt{N}}\right)\right|\\\nonumber
&=&\sup_{x\in\mathbb{R}}\sum_k\lambda_k\left|P_k\left[\frac{1}{N}X^{(N)}\leq \frac{\Sigma_{N,k}}{\sqrt{N}}x+\mu_{N,k}\right]-\Phi(x)\right|\\\nonumber
&=&\sup_{x\in\mathbb{R}}\sum_k\lambda_k\left|P_k\left[\frac{X^{(N)}-N\mu_{N,k}}{\sqrt N\Sigma_{N,k}}\leq x\right]-\Phi(x)\right|\\\nonumber
&=&\sup_{x\in\mathbb{R}}\sum_k\lambda_k\left|P_k\left[S_{N,k}/\Sigma_{N,k}\leq x\right]-\Phi(x)\right|\\\nonumber
&\leq&\sum_k\lambda_k\frac{C_0M}{\sigma_{-}^3\sqrt N}\\\nonumber
&=&\frac{C_0M}{\sigma_{-}^3\sqrt N}.
\end{eqnarray}
The last inequality follows from (\ref{BEbound}).}
\end{proof}

Note that the bound (\ref{RTheorem}) does not depend on any structure/symmetry of the underlying input state. This is in contrast to the previous de Finetti-type representation theorems that heavily rely on symmetry, such as the permutational invariance.

%However, due to the possible high asymmetry of the underlying state, the inequality (\ref{RTheorem})  cannot provide the convergence to the mixture of normal distributions. It only says that the distribution of relative frequency is highly peaked at one of the mean averages $x_{N,k}$ with the probability $\lambda_k$. However, if the convergence of $x_{N,k}\rightarrow x_k$ is provided for all separable components, we have $f_N(x)\rightarrow\sum_k\lambda_k\delta(x-x_k)$ where $f_N(x)$ is the probability density function (PDF) of the relative frequency, and $\delta(x)$ denotes the Dirac delta function.

\emph{Convergence rates and quantum game.}$-$ In this section we will show that entangled states can behave very differently in certain cases compared to separable states with the respect to the distribution of the relative frequency. To illustrate our findings we will define the problem as an information-theoretic game between two players, Alice and Bob.

Suppose that Alice performs some POVM and generates a random variable $X\in\{x_1,x_2,\dots\}$ is strictly unsharp, i.e. $\Var[x]_{\rho}\geq\sigma_{-}>0$ for all $\rho$. As previously, we assume that third moments are bounded by $M>0$. She asks Bob to supply her with inputs, and his goal is to make the relative frequency $R_N=\frac{1}{N}(X_1+\dots+ X_N)$ as close as possible to some pre-defined value $X_c$. More precisely, he will try to maximize the probability
\begin{equation}\label{win P}
P_{N}=P\left[|R_N-X_c|\leq\frac{\epsilon}{N^{\alpha}}\right],
\end{equation}
with $\epsilon,\alpha>0$ being fixed parameters. The parameter $\alpha$ quantifies the convergence rate of the relative frequency to the constant $X_c$. Our goal here is to show that the probability $P_N$ is negligible whenever $\alpha>1/2$ for all separable states. And indeed, the bound (\ref{RTheorem}) states that the distribution of $R_N$ is a mixture of Gaussians, therefore the error (as quantified by the convergence rate) cannot scale better than $1/\sqrt{N}$. We fix $\alpha>1/2$.
\begin{theorem}
\begin{equation}\label{SepBound}
P_{N}\leq\sqrt{\frac{2}{\pi}}\frac{\epsilon}{\sigma_{-}}\frac{1}{N^{\alpha-\frac{1}{2}}}+\frac{2C_0M}{\sigma_{-}^3\sqrt N}
\end{equation}
for all separable inputs.
\end{theorem}
\begin{proof}{\footnotesize For a separable input $\rho^{(N)}=\sum_k\lambda_k\rho_k^{(N)}$ we have $P_N=\sum_k\lambda_k P_{N,k}$. Therefore it is sufficient to prove (\ref{SepBound}) for a product state. We set $\rho^{(N)}=\rho_1\otimes\dots\otimes\rho_N$ and, as previously $\Sigma_N^2=\frac{1}{N}\sum_{i=1}^N\sigma_i^2$. We have
%\begin{widetext}
\begin{eqnarray}\nonumber
P_{N}&=&P\left[|R_N-X_c|\leq\frac{\epsilon}{N^{\alpha}}\right]\\\nonumber
&=&P\left[X_c-\frac{\epsilon}{N^{\alpha}}\leq\frac{1}{N}X^{(N)}\leq X_c+\frac{\epsilon}{N^{\alpha}}\right]\\\nonumber
%&=&P\left[X_c-\frac{1}{N}\avg{X^{(N)}}-\frac{\epsilon}{N^{\alpha}}\leq\frac{1}{N}(X^{(N)}-\avg{X^{(N)}})\leq X_c-\frac{1}{N}\avg{X^{(N)}}+\frac{\epsilon}{N^{\alpha}}\right]\\\nonumber
&=&P\left[A_N-a_N\leq S_N/\Sigma_N\leq A_N+a_N\right]\\\nonumber
&\leq&\Phi(A_N+a_N)-\Phi(A_N-a_N)+\frac{2C_0M}{\sigma_{-}^3\sqrt N},
\end{eqnarray}
%\end{widetext}
where $A_N=\frac{\sqrt{N}}{\Sigma_N}(X_c-\frac{1}{N}\avg{X^{(N)}})$ and $a_N=\frac{\epsilon}{\Sigma_N N^{\alpha-\frac{1}{2}}}$. The last inequality follows from the Berry-Esseen bound (\ref{BEbound}). For $a>0$ the function $\Phi(x+a)-\Phi(x-a)$ reaches its absolute maximum for $x=0$, hence $\Phi(x+a)-\Phi(x-a)\leq \Phi(a)-\Phi(-a)=2\Phi(a)-1$. Here, we used $\Phi(x)+\Phi(-x)=1$. Furthermore, the function $\Phi(x)$ is concave for $x\geq0$, therefore $\Phi(x)\leq\frac{1}{2}+\frac{x}{\sqrt{2\pi}}$. Finally, we have
\begin{eqnarray}\nonumber
P_N&\leq&\Phi(A_N+a_N)-\Phi(A_N-a_N)+\frac{2C_0M}{\sigma_{-}^3\sqrt N}\\\nonumber
&\leq&2\Phi(a_N)-1+\frac{2C_0M}{\sigma_{-}^3\sqrt N}\\\nonumber
&\leq&\sqrt{\frac{2}{\pi}}a_N+\frac{2C_0M}{\sigma_{-}^3\sqrt N}\\\nonumber
&=&\sqrt{\frac{2}{\pi}}\frac{\epsilon}{\Sigma_N N^{\alpha-\frac{1}{2}}}+\frac{2C_0M}{\sigma_{-}^3\sqrt N}\\\nonumber
&\leq&\sqrt{\frac{2}{\pi}}\frac{\epsilon}{\sigma_{-}}\frac{1}{N^{\alpha-\frac{1}{2}}}+\frac{2C_0M}{\sigma_{-}^3\sqrt N}\\\nonumber
\end{eqnarray}}
\end{proof}

The bound (\ref{SepBound}) states that the wining probability vanishes asymptotically $P_N\rightarrow0$ with $N\rightarrow+\infty$, for all $\alpha>1/2$. Therefore, Bob will fail to win the game with certainty by using separable inputs. Now we will provide a simple example where entanglement is able to beat the bound given by (\ref{SepBound}).

\emph{Entanglement example.}$-$ Consider a qubit three-outcome POVM with the elements on ``equilateral triangle'' $E_i=\frac{1}{3}( \openone+\vec{m}_i\cdot\vec{\sigma})$, where $\vec{m}_0=(1,0,0)^{T}$, $\vec{m}_{\pm1}=(-1/2,0,\pm\sqrt3/2)^{T}$ and $\vec{\sigma}=\{\sigma_x,\sigma_y,\sigma_z\}$ is the vector of three Pauli matrices. We define the corresponding random variable with three possible values $X \in\{-1,0,1\}$ and we set $X_c=0$. It is convenient to introduce two operators $A=\sum_ix_iE_i=-E_{-1}+E_{1}=\frac{1}{\sqrt3}\sigma_z$ and $B=\sum_ix_i^2E_i=E_{-1}+E_{1}=\frac{2}{3}\openone-\frac{1}{3}\sigma_x$. We have $\Var_{\rho}[X]=\avg{B}_{\rho}-\avg{A}_{\rho}^2=\frac{2}{3}-\frac{x}{3}+\frac{z^2}{3}$, where $x$ and $z$ are components of the Bloch vector of the state $\rho$. Clearly $x^2+z^2\leq1$. A simple calculation shows that $\Var_{\rho}[X]\geq\frac{1}{3}$, hence $\sigma_{-}=\frac{1}{3}$. Furthermore $|X|\leq1$, thus the third moment is bounded and we have $M=1$. The bound (\ref{SepBound}) applies to all separable inputs and $\alpha>1/2$.

On the other hand, let Bob use the following input state
\begin{equation}\label{Ansatz}
\ket{\psi}=\frac{1}{\sqrt{2L+1}}\sum_{m=-L}^{L}\ket{J,m},
\end{equation}
where we set $L=N^{\beta}$ with $0<\beta<1/2$. Here, we use the spin-$J$ representation for $N$-qubit permutational invariant pure state, i.e. any state can be written as $\sum_{m=-J/2}^{J/2}c_m\ket{J,m}$, with $J=N/2$ and $\ket{J,J}=\ket{1}^{\otimes N}$. The state (\ref{Ansatz}) is very closed to the Dicke-squeezed state~\cite{Duan} introduced for the purposes of quantum metrology. Clearly, the mean value $\avg{X^{(N)}}=\bra{\psi}\sum_{i=1}^{N}A_i\ket{\psi}=\frac{2}{\sqrt3}\bra{\psi}S_z\ket{\psi}=0$, where $S_z=\frac{1}{2}\sum_i\sigma_{z,i}$ is the total spin operator along $z$-direction. Keeping in mind that $\avg{X^{(N)}}=X_c=0$, we can lower-bound the winning probability by using the Chebyshev's inequality
\begin{eqnarray}\label{Chebyshev}
P_N&=&P\left[|\frac{1}{N}X^{(N)}-X_c|\leq\frac{\epsilon}{N^{\alpha}}\right]\\\nonumber
&=&P\left[|X^{(N)}-\avg{X^{(N)}}|\leq\frac{\epsilon}{N^{\alpha-1}}\right]\\\nonumber
&\geq&1-\frac{N^{2(\alpha-1)}}{\epsilon^2}\Var[X^{(N)}]\\
&=&1-s_N,
\end{eqnarray}
where $s_N=\frac{N^{2(\alpha-1)}}{\epsilon^2}\Var[X^{(N)}]$ upper-bounds the probability of failure. Our goal is to show that $s_N$ is negligible for $N$ large. A  simple calculation shows that $\Var[X^{(N)}]=\frac{N}{3}-\frac{2}{3}\bra{\psi}S_x\ket{\psi}+\frac{4}{3}\Delta S_{z}^{2}$. Firstly, we calculate $\Delta S_{z}^{2}$ directly by substituting (\ref{Ansatz})
\begin{eqnarray}\nonumber
\Delta S_{z}^{2}&=&\frac{1}{2L+1}\sum_{m=-L}^{L}m^2\\
&=&\frac{1}{3} L(L+1).
\end{eqnarray}
The wavefunction $\ket{\psi}$ is real (with the respect to basis $\ket{J,m}$), hence $\bra{\psi}S_x\ket{\psi}=\frac{1}{2}\bra{\psi}(S_{-}+S_{+})\ket{\psi}=\bra{\psi}S_{-}\ket{\psi}$, where $S_{-}$ is the spin-ladder operator\\ $S_{-}\ket{J,m}=\sqrt{J(J+1)-m(m-1)}\ket{J,m-1}$. We have
\begin{eqnarray}\nonumber
\bra{\psi}S_{-}\ket{\psi}&=&\frac{1}{2L+1}\sum_{m=-L+1}^{L}\sqrt{J(J+1)-m(m-1)}\\\nonumber
&=&\frac{\sqrt{J(J+1)}}{2L+1}\sum_{m=-L+1}^{L}\left(1-\frac{m(m-1)}{J(J+1)}\right)^{\frac{1}{2}}\\\nonumber
&\geq&\frac{J}{2L+1}\sum_{m=-L+1}^{L}\left(1-\frac{m(m-1)}{J(J+1)}\right)^{\frac{1}{2}}\\\nonumber
&\geq&\frac{J}{2L+1}\sum_{m=-L+1}^{L}\left(1-\frac{m(m-1)}{J(J+1)}\right)\\\nonumber
&=&\frac{J}{2L+1}\left(2L-\frac{2L(L^2-1)}{3J(J+1)}\right)\\\nonumber
&=&\frac{N L}{2L+1}-\frac{4L(L^2-1)}{3(2L+1)(N+2)}\\
&\geq&\frac{N L}{2L+1}-\frac{2L^2}{3(N+2)},
\end{eqnarray}
where we used $-\frac{L(L^2-1)}{2L+1}\geq-\frac{L^2}{2}$ for $L\geq0$. The second inequality follows from concavity of $\sqrt{1-x}$ for $0\leq x\leq1$. Now we can derive the bound for variance
\begin{eqnarray}
\Var[X^{(N)}]&=&\frac{N}{3}-\frac{2}{3}\bra{\psi}S_x\ket{\psi}+\frac{4}{3}\Delta S_{z}^{2}\\\nonumber
&=&\frac{N}{3}-\frac{2}{3}\bra{\psi}S_x\ket{\psi}+\frac{4}{9} L (L+1)\\\nonumber
&\leq&\frac{N}{3}+\frac{4}{9} L (L+1)-\frac{2}{3}\frac{N L}{2L+1}+\frac{4L^2}{9(N+2)}\\\nonumber
&=&\frac{4}{9} L (L+1)+\frac{N}{3(2L+1)}+\frac{4L^2}{9(N+2)}=Q.
\end{eqnarray}
For $L=N^{\beta}$ and $N$ being large, the right-hand side of the last inequality scales as $Q\sim\frac{4}{9}N^{2\beta}+\frac{1}{6}N^{1-\beta}+\frac{4}{9}N^{2\beta-1}$. Since $\beta<1/2$ the last therm is negligible. Furthermore, we see that the best rate is achieved for $\beta=1/3$. Finally we get the estimation for the maximal error $s_N\leq\frac{N^{2(\alpha-1)}}{\epsilon^2}Q\sim\frac{11N^{2\alpha-4/3}}{18\epsilon^2}$ which is negligible for all $\alpha<2/3$.

\emph{Concluding remarks.}$-$
In this letter we provided the central limit representation theorem for all separable states subjected to unsharp quantum measurements. We have shown that errors and convergence rates of the relative frequency are at best at the order of $1/\sqrt{N}$. This scaling factor comes exclusively due to the measurement uncertainty. On the other hand, we have shown that the better rates are achievable by using entangled inputs. This means that in certain cases entanglement can ``boost'' the convergence rates and suppress the measurement errors. Thus, our findings can be potentially used for quantum metrology purposes~\cite{Lloyd}. In addition, the framework developed here can be used for probabilistic entanglement detection. Namely, we have shown that the wining probability \eqref{win P} asymptotically reaches 1 for certain entangled states, whereas it asymptotically vanishes for all separable states, thus one can verify the presence of entanglement with a high probability even by using a single-copy of a target state (provided that $N$ is sufficiently large). This is in agreement with our recent conclusions on ``single-copy entanglement detection'' presented in \cite{DD17}.

\emph{Acknowledgments.}$-$The authors thank Nata\v sa Dragovi\'c and \v{C}aslav Brukner for helpful
comments and acknowledge support from the European Commission through the projects RAQUEL (No.\ 323970) and Serbian Ministry of Science (Project ON171035).

\end{document}